\documentclass[12pt,reqno]{article}

\usepackage[english]{babel}
\usepackage{amsfonts, amssymb, amsmath, amsthm} 

\usepackage{eucal}
\usepackage{nicefrac}
\usepackage{color}
\usepackage{mathrsfs}
\usepackage{graphicx}

\usepackage[letterpaper,top=2cm,bottom=2cm,left=3cm,right=3cm,marginparwidth=1.75cm]{geometry}

\usepackage[utf8]{inputenc}

\usepackage{mathtools}

\parskip=\medskipamount
\newtheorem{definition}{\noindent \noindent {\bf
Definition}}[section]

\newtheorem{prop}{{\bf Proposition}}[section]
\newtheorem{corollary}{{\bf Corollary}}[section]
\newtheorem{remark}{{\bf Remark}}[section]
\newtheorem{example}{{\bf Example}}[section]

\usepackage{hyperref}
\hypersetup{colorlinks=true, allcolors=blue}
\newcommand{\pder}[2]{ \frac{\partial #1}{\partial #2} }
\newcommand{\R}{\mathbb{R}}

\title{A geometric description of some thermodynamical systems}
\author{\sffamily
Manuel de León$^{1,2}$, \thanks{mdeleon@icmat.es\qquad\qquad ORCID: 0000-0002-8028-2348}\, 
Jaime Bajo$^{3}$, \thanks{jaime.bajo@estudiantes.uva.es\qquad\qquad ORCID: 0000-0001-8665-7713}\,
\\[1ex]
\\[0.1ex]
\normalsize\itshape\sffamily
$^1$Instituto de Ciencias Matemáticas (CSIC), Madrid, Spain
\\[0.1ex]
\normalsize\itshape\sffamily
$^2$Real Academia de Ciencias, Madrid, Spain
\\[0.1ex]
\normalsize\itshape\sffamily
$^3$Universidad de Valladolid, Spain
\\[0.1ex]}

\begin{document}
\maketitle

\begin{abstract}
In this paper we show how almost cosymplectic structures are a natural framework to study thermodynamical systems. Indeed, we are able to obtain the same evolution equations obtained previously by Gay-Balmaz and Yoshimura \cite{hiro} using variational arguments. The proposed geometric description allows us to apply geometrical tools to discuss reduction by symmetries, the Hamilton-Jacobi equation or discretization of these systems.
 
\end{abstract}
\tableofcontents
\section{Introduction}

Just as symplectic geometry is the natural setting for developing time-independent mechanics, so is cosymplectic geometry in the case of time dependence \cite{albert1989theoreme,cantrijn1992gradient,de2011methods,libermann2012symplectic}. On the other hand, contact geometry is the arena for studying systems with dissipation (in more precise terms, those whose Lagrangian depends on the action itself) \cite{de2019contact}.

Traditionally, the geometry of equilibrium thermodynamics has been mainly studied via contact geometry in terms of contact manifolds \cite{Bravetti2018a,simoes2020contact,elotro} (see also \cite{mrugala1,mrugala2,lacomba1,lacomba2}); in this geometric setting, thermodynamic properties are encoded by Legendre submanifolds of the thermodynamic phase space. However, in this paper we will take a different view, and see how different thermodynamic systems can be described by (almost) cosymplectic structures, that, in some sense, could be considered as natural extensions of contact geometry, even if they exhibit very different features. 
In this way, we reobtain the evolution equations obtained in a recent survey by Gay-Balmaz and Yoshimura \cite{hiro}. In that survey, the authors have obtained these equations using a variational approach of nonequilibrium thermodynamics for the finite-dimensional case of discrete systems, as well as for the infinite-dimensional case of continuum systems.

Our plan is to follow the same scheme that in \cite{hiro}. So, we first consider adiabatically closed simple systems; then, we consider adiabatically closed non-simple systems; and, further, we consider the geometric formulation for open systems. In this approach, we gradually increase the level of complexity. Indeed, we start by studying an adiabatically closed system that has only one entropy variable or, equivalently, one temperature. Such systems, called simple systems, may involve the irreversible processes of mechanical friction and internal matter transfer. A more general class are adiabatically closed thermodynamic systems with several entropy variables, which may also involve the irreversible process of heat conduction. Another further step is to consider open thermodynamic systems, which can exchange heat and matter with the exterior. 

The paper is structured as follows. In Section 2 we recall the main notions and results concerning the coysmplectic formulation of time-dependent mechanics, in the Lagrangian and the Hamiltonian descriptions, both related by the Legendre transformation. We also include the notion of semibasic form, which is the geometric notion corresponding to external forces acting on the system. Section 3 is devoted to introduce some extensions and generalizations of almost cosymplectic structures which will be used in the rest of the paper. So, in Section 4 we apply the previous definitions to several cases of thermodynamical systems, providing a geometric setting for Adiabatically Closed Simple Thermodynamic Systems, Systems with Internal Mass Transfer, Adiabatically Closed Non-Simple Thermodynamic Systems and Open Simple Thermodynamic Systems. Finally, in Section 5 we present some conclusions as well as some future lines of research.

\section{Dynamics on cosymplectic geometry}

\subsection{Cosymplectic Hamiltonian formalism}

A cosymplectic structure on an odd-dimensional manifold $M$ is a pair $(\Omega, \eta)$ where
$\Omega$ is a closed 2-form, $\eta$ is a closed 1-form, and
$\Omega^n \wedge\eta \not= 0$, where $M$ has dimension $2n+1$.
$(M, \Omega, \eta)$ will be called a cosymplectic manifold.

There is a Darboux theorem for a cosymplectic manifold, that is,
there are local coordinates (called Darboux coordinates) $(q^i, p_i, z)$ around each point of $M$
such that
$$
\Omega = dq^i \wedge dp_i \; , \; \eta = dz
$$
There exists a unique vector field (called Reeb vector field) $\mathcal R$ such that
$$
i_{\mathcal R} \, \Omega = 0 \; , \; i_{\mathcal R}\, \eta = 1
$$
In Darboux coordinates we have
$$
\mathcal R = \frac{\partial}{\partial z}
$$

Let $H : M \longrightarrow \mathbb R$ be a Hamiltonian function, say $H = H(q^i, p_i, z)$.

Consider the vector bundle isomorphism
$$
\flat : TM \longrightarrow T^* M \; , \; \flat(v) = i_v \, \Omega + \eta (v) \, \eta
$$
and define the gradient of $H$ by
$$
\flat({\rm grad} \; H) = dH
$$
Then
\begin{equation}\label{hcosymp}
{\rm grad} \; H = \frac{\partial H}{\partial p_i} \frac{\partial}{\partial q^i} - 
\frac{\partial H}{\partial q^i} \frac{\partial}{\partial p_i} + \frac{\partial H}{\partial z} \, \frac{\partial}{\partial z}
\end{equation}

Next we can define two more vector fields:

\begin{itemize}

\item The Hamiltonian vector field 
$$
X_H = {\rm grad} \; H - \mathcal R (H) \mathcal R
$$

\item and the evolution vector field
$$
{\mathcal E}_H = X_H + {\mathcal R}
$$
\end{itemize}

From (\ref{hcosymp}) we obtain the local expression

\begin{equation}\label{hcosymp2}
{\mathcal E}_H = \frac{\partial H}{\partial p_i} \frac{\partial}{\partial q^i} - 
\frac{\partial H}{\partial q^i} \frac{\partial}{\partial p_i} +  \frac{\partial}{\partial z}
\end{equation}
Therefore, an integral curve $(q^i(t), p_i(t), z(t))$ of ${\mathcal E}_H$ satisfies the 
time-dependent Hamilton equations

\begin{eqnarray}\label{hsympl3}
\frac{dq^i}{dt} & = & \frac{\partial H}{\partial p_i} \\
\frac{dp_i}{dt} & = & - \frac{\partial H}{\partial q^i}\\
\frac{dz}{dt} & = & 1
\end{eqnarray}
and then $z=t+const$ so that both coordinates can be identified.

\subsubsection{Cosymplectic Hamiltonian formalism on extended cotangent bundles}

Now, we consider a time-dependent Hamiltonian function $H = H(q^i, p_i, t)$, that is, $H$ is a function defined on the so-called extended cotangent bundle $T^*Q \times \mathbb R$, where $Q$ is the configuration manifold.

On $T^*Q$ there exists a canonical 1-form $\theta_Q$ (the so-called Liouville form) given by
$$
\theta_{\alpha_q} (X) = \langle \alpha_q , T\pi_Q(X) \rangle ,
$$
where $\alpha_q \in T_q^*Q$, $X \in T_{\alpha_q}(T^*Q)$, and $\pi_Q : T^*Q \longrightarrow Q$ is the canonical projection.

In bundle coordinates $(q^i, p_i)$ on $T^*Q$, we have
$$
\theta_Q = p_i dq^i
$$
So, $\omega_Q = - d\theta_Q$ is a symplectic form with local expression
$$
\omega_Q = dq^i \wedge dp_i
$$

\begin{remark} {\rm The form $\omega_Q$ is called the canonical symplectic form on $T^*Q$. Indeed, Darboux theorem states that any symplectic form is locally equivalent to a canonical symplectic form.}
\end{remark}

Then, we can consider the extended cotangent bundle $T^*Q \times \mathbb R$, equipped with the 2-form $\omega_Q$ (that is, the pull-back via the canonical projection $\pi: T^*Q \times \mathbb R \longrightarrow T^*Q$) and the 1-form $dz$, where $z$ is the canonical coordinate in $\mathbb R$. Then, the pair $(\omega_Q, dz)$ is a cosymplectic structure on $T^*Q \times \mathbb R$.

Now, given a Hamiltonian function $H : T^*Q \times \mathbb R \longrightarrow \mathbb R$, we can develop the Hamiltonian formalism just repeating the notions of the previous section.

\subsection{Cosymplectic Lagrangian formalism}

We will recall here the geometric formalism for time-dependent Lagrangian systems.
In this case, we also have a Lagrangian $L : TQ \times \mathbb R \longrightarrow \mathbb R$, and we will consider the cosymplectic structure given by the pair
$(\Omega_L, dz)$, $z$ being a global coordinate in $\mathbb R$, and where:
$$
\Omega_L = - d \lambda_L 
$$
$\lambda_L$ being the 1-form on $TQ \times \mathbb R$ defined by
$$
\lambda_L = S^*(dL)
$$
Here, $S$ is the canonical vertical endomorphism (or almost tangent structure) defined on $TQ$ but considered now acting on $TQ \times \mathbb R$ in the obvious way. Recall that in bundle coordinates $(q^i, \dot{q}^i)$ on $TQ$ we have
$$
S = \frac{\partial}{\partial \dot{q}^i} \otimes dq^i 
$$
Therefore, we obtain
$$
\lambda_L = \frac{\partial L}{\partial \dot{q}^i} \, dq^i
$$
A Lagrangian $L$ is said to be regular if and only if the matrix
$$
(\frac{\partial^2L}{\partial \dot{q}^i \partial \dot{q}^j})
$$
is non-singular.

The Lagrangian energy $E_L$ is defined as
$$
E_L = \Delta(L) - L,
$$
where $\displaystyle{\Delta = \dot{q}^i \frac{\partial}{\partial \dot{q}^i}}$ is the Liouville vector field on $TQ$ (here, considered on $TQ \times \mathbb R$ in the obvious manner).

It is easy to check that, indeed, if $L$ is regular then
$$
\Omega_L^n \wedge  dz \not= 0,
$$
and, conversely. Again, we have a Reeb vector field 
$$
\mathcal R_L = \frac{\partial}{\partial z} - W^{ij} \frac{\partial^2 L}{\partial \dot{q}^j \partial z} \, \frac{\partial}{\partial \dot{q}^i}
$$

Consider now the following vector fields determined by means of the vector bundle isomorphism
\begin{eqnarray*}
&&\flat_L  :  T(TQ \times \mathbb R) \longrightarrow T^*(TQ \times \mathbb R)\\
&&\flat_L (v) = i_v \, \Omega_L + dz (v) \, dz
\end{eqnarray*}
say,
\begin{enumerate}
\item the gradient vector field
$$
{\rm grad} \; (E_L) = \sharp_L (dE_L)
$$
\item the Hamiltonian vector field
$$
X_{E_L} = {\rm grad} \; E_L - \mathcal R_L (E_L) \, \mathcal R_L
$$
\item and the evolution vector field
$$
{\mathcal E}_L = X_{E_L} + \mathcal R_L
$$
\end{enumerate}
where $\sharp_L = (\flat_L)^{-1}$ is the inverse of $\flat_L$.
(For the sake of simplicity, here and in the following, we denote
${\mathcal E}_L$ instead of ${\mathcal E}_{E_L}$).

The evolution vector field ${\mathcal E}_L$ is locally given by
\begin{equation}\label{cosylagr1}
{\mathcal E}_ L = \dot{q}^i \, \frac{\partial}{\partial q^i} + B^i \, \frac{\partial}{\partial \dot{q}^i} + \frac{\partial}{\partial z}
\end{equation}
where
\begin{equation}\label{cosylagr2}
B^i \, \frac{\partial}{\partial \dot{q}^i}(\frac{\partial L}{\partial \dot{q}^j}) + \dot{q}^i \, 
\frac{\partial}{\partial q^i}(\frac{\partial L}{\partial \dot{q}^j}) - \frac{\partial L}{\partial q^j} = 0
\end{equation}
Now, if $(q^i(t), \dot{q}^i(t), z(t))$ is an integral curve of ${\mathcal E}_L$ then it satisfies the 
usual Euler-Lagrange equations
\begin{equation}\label{cosylagr3}
\frac{d}{dt} (\frac{\partial L}{\partial \dot{q}^i}) - \frac{\partial L}{\partial q^i} = 0
\end{equation}
since $z = t + constant$. 

\subsection{The Legendre transformation}

Assume that $L : TQ \times \mathbb{R}$ is a time-dependent Lagrangian. Then, one can define the Legendre transformation
$$
Leg : TQ \times \mathbb{R} \longrightarrow T^*Q \times \mathbb{R}
$$
as the mapping given in local coordinates by
$$
Leg(q^i, \dot{q}^i, z) = (q^i, p_i, z) 
$$
where $\displaystyle{p_i= \frac{\partial L}{\partial \dot{q}^i}}$. Of course, one can give a global definition, independent of the chosen local coordinates (see \cite{de2011methods}).

One can easily prove that $Leg$ is a local diffeomorphism if and only if the Lagrangian $L$ is regular. In addition, $L$ is said to be hyperregular if $Leg$ is a global diffeomorphism.

In that case, if $E_L$ is the Lagrangian energy, one can define a Hamiltonian energy $H$ on $T^*Q \times \mathbb{R}$ by
$$
H=E_L\circ Leg^{-1}
$$

A simple computation shows that
$$
Leg^* \theta_Q = \lambda_L
$$
so that the Legendre transformation preserves the cosymplectic structures $(\Omega_L, dz)$ and $(\Omega_Q, dz)$; in other words, it is a cosymplectomorphism.

Therefore, due to the above relations between the energies, one can deduce that the gradient, Hamiltonian and evolution vector fields are related by the Legendre transformation.

\subsection{Forces and semibasic forms}

A force on a mechanical system with configuration manifold $Q$ is interpreted as a semibasic 1-form on the tangent bundle $TQ$, or, alternatively, on the cotangent bundle $T^*Q$. Let us recall that a semibasic form on $TQ$ (resp., on $T^*Q$) is a 1-form $\tilde{\alpha}$ (resp. $\alpha$) on $TQ$ (resp., on $T^*Q$) such that it vanishes acting on vertical vectors. So, the local representations are
\begin{eqnarray*}
\tilde{\alpha} &=& \tilde{\alpha}_i(q, \dot{q}) dq^i \\
(\hbox{resp.} \; 
\alpha& = & \alpha_i(q, p) dq^i )
\end{eqnarray*}

Similar notions can be considered for time-dependent Lagrangian and Hamiltonian systems. In this case, we consider semibasic forms dependent on time, that is, 1-forms on $TQ \times \mathbb R$ (resp. on $T^*Q \times \mathbb{R}$) such that it vanishes acting on vertical vectors with respect to the fibration $TQ \times \mathbb R \longrightarrow Q \times \mathbb R$ (resp., $T^*Q \times \mathbb R \longrightarrow Q \times \mathbb R$).

This means that the local expressions are
\begin{eqnarray*}
\tilde{\alpha} &=& \tilde{\alpha}_i(q, \dot{q}, z) dq^i \\
(\hbox{resp.} \; 
\alpha & = & \alpha_i(q, p, z) dq^i )
\end{eqnarray*}

The way to include external forces in the cosymplectic formulation of mechanics is just as follows:

For Lagrangian mechanics, we consider the equation
$$
\flat_L (X) = i_X \Omega_L + dz(X) dz = \tilde{\alpha}
$$
or, for Hamiltonian mechanics
$$
\flat (X) = i_X \omega_Q + dz(X) dz = \alpha
$$
Since the Legendre transformation preserves the fibers, one deduces that
$$
Leg^* \alpha = \tilde{\alpha}
$$

\section{Almost cosymplectic structures}

\subsection{Partially cosymplectic structures} 

\begin{definition}
An almost cosymplectic structure on a manifold $M$ is a pair $(\omega, \eta)$, where $\omega$ is a 2-form and $\eta$ a 1-form such that
$$
\omega^n \wedge \eta \not= 0.
$$
(Here, $2n+1$ is the dimension of $M$).
When $\omega$ and $\eta$ are both closed, the structure is called cosymplectic.  
\end{definition}

Along this paper, we will consider a particular kind of almost cosymplectic structures $(\omega, \eta)$, those whose 2-form $\omega$ is closed but whose 1-form $\eta$ is not. We call these structures as {\bf partially cosymplectic}.

\begin{example}
    Notice that if $(M,\eta)$ is a contact manifold \cite{de2019contact} of dimension $2n+1$, in particular, $(d\eta,\eta)$ define a partially cosymplectic structure on $M$.
\end{example}

Assume that $(\omega, \eta)$ is a partially cosymplectic manifold. We denote by
$$
K = \ker \eta
$$
Then, $K$ is a vector subbundle of the tangent bundle $TM$, even a symplectic vector subbundle when we consider the restriction of the 2-form $\omega$ to its fibers. However, considered as a distribution on $M$, it is not involutive. Indeed, if $X$ and $Y$ are vector fields in $K$ (we are assuming some abuse of language), we have
$$
d \eta (X, Y) = X(\eta(Y)) - Y (\eta(X)) - \eta([X,Y]) = - \eta([X,Y]).
$$
We also have in this context the notion of Reeb vector field.

\begin{prop}
Given a partially cosymplectic structure $(\omega, \eta)$ on $M$, there exists a unique vector field ${\cal R}$ such that
$$
i_{\cal R} \, \omega = 0 \qquad, \qquad i_{\cal R} \, \eta = 1.
$$

${\cal R}$ will be called the Reeb vector field.
\end{prop}
\begin{proof}
As $(M,\omega)$ is a pre-symplectic manifold of corank 1, according to the generalized Darboux Theorem \cite{de2011methods} (see also \cite{godbillon1969geometrie}) for each $x \in M$ there exists a coordinate neighborhood $U_x$ with local coordinates $(q^1,\cdots,q^r,p_1,\cdots,p_r,s)$ such that:
$$
\omega = dq^i\wedge dp_i
$$
In particular, $i_\pder{}{s}\omega=0$. Thus, it must be $\eta\left(\pder{}{s}\right)\neq 0$ for each point of $U_x$, since $\omega^n\wedge \eta \neq 0$. Setting $\mathcal{R}_x=\frac{1}{\eta\left(\pder{}{s}\right)}\pder{}{s}$, then $\mathcal{R}_x$ is a local Reeb vector field. 

Consider a partition of unity on $M$ subordinate to the atlas $\{U_x\}$, $\{(U_i,f_i)\}$. For each $i$ there exists $x_i$ such that $U_i\subseteq U_x$. Let $\mathcal{R}=\sum_i f_i\mathcal{R}_{x_i}$. Then $i_{\mathcal{R}}\omega=0$ and $i_{\mathcal{R}}\eta=1$.

Let $\mathcal{R}_1$, $\mathcal{R}_2$ be two Reeb vector fields. Consider $X=\mathcal{R}_1-\mathcal{R}_2$. If for some $x\in M$, $X(x)\neq 0$ then we can extend the tangent vector into a basis of $T_xM$, $\{X(x),e_1,\cdots,e_{2n}\}$. Then we would have that $(\omega^n\wedge\eta) (X(x),e_1,\cdots,e_{2n})=0$. Thus, it must be $X = 0$ and the Reeb vector must be unique.
\end{proof}
\begin{corollary} \label{CotangentDecomposition}
    Let $(\omega,\eta)$ be a partially cosymplectic structure over $M$. Consider $\gamma:TM\longrightarrow T^*M$ the morphism defined by $\gamma(X)=i_X\omega$ and let $H=Im(\gamma)$. Then we can express the cotangent bundle as the following Whitney sum:
    $$
    T^*M=H\oplus\langle\eta\rangle
    $$
\end{corollary}
Given a partially cosymplectic structure $(\omega, \eta)$ on $M$, we can introduce the notion of evolution vector field for any function $f$ defined on $M$; indeed, given a function $f \in C^\infty(M)$, there is a unique vector field $E_f$, which will be called evolution vector field, such that
$$
\flat(E_f) = df + \eta,
$$
where $\flat : TM \longrightarrow T^*M$ is the isomorphism defined by
$$
\flat(X) = i_X \, \omega + \eta(X) \eta
$$

\begin{remark}
{\rm 
Our structure does not coincide with the one previously defined by Acakpo \cite{acakpo2022stable}, called stable Hamiltonian structure.
Indeed, a \textbf{stable Hamiltonian structure} (SHS) is a triple $(M, \omega, \lambda)$ where $M$ is a $2n + 1$ dimensional manifold, $\omega$ is a closed $2$-form and $\lambda$ is a $1$-form such that 
$$
\lambda \wedge \omega^n \neq 0, \,\,\, \ker \omega \subseteq \ker d\lambda.
$$
There are also some relations with the mechanical presymplectic structures defined in \cite{jcm}.
}
\end{remark}

\subsection{Partially cosymplectic structures of higher order}

In order to obtain an appropriate framework for more complex thermodynamical systems, we introduce a generalization of the above geometric structures.

\begin{definition}
An almost cosymplectic structure of order $p$ on a manifold $M$ is a $(p+1)$-tuple $(\omega, \eta_1, \cdots , \eta_p)$, where $\omega$ is a 2-form and $\eta_1, \cdots, \eta_p$ are 1-forms such that
$$
\omega^n \wedge \eta_1 \wedge \cdots \wedge \eta_p \not= 0.
$$
(Here, $2n+p$ is the dimension of $M$).
When $\omega$ and $\eta$ are both closed, the structure is called cosymplectic of order $p$. If only $\omega$ is closed, then $(\omega, \eta_1, \cdots , \eta_p)$ is called partially cosymplectic of order $p$.
\end{definition}

\begin{prop}
    Given an almost cosymplectic structure of order $p$ $(\omega,\eta_1,\cdots,\eta_p)$ on $M$, the $\mathcal{C}^{\infty}(M)$-morphism given by:
    \begin{align*}
    \flat:&TM \longrightarrow T^*M \\
          &X \quad\longmapsto i_X\omega+\sum_k\eta_k(X)\eta_k
    \end{align*}
    is an isomorphism of $\mathcal{C}^{\infty}(M)$ modules.
\end{prop}
\begin{proof}
    It is enough to prove that $\flat_x:T_xM\longrightarrow T_x^*M$ is an isomorphism and so we can argue locally.

    Since both $T^*_xM$ and $T_xM$ have equal dimensions, it suffices to show that $\flat$ is one to one. Suppose there is a non zero tangent vector $X \in T_xM$ such that $\flat(X)=0$. 

    Then $\flat(X)(X)=\sum_k\left(\eta_k(X)\right)^2=0$, so we can conclude that $\eta_k(X)=0 \quad \forall \, k$.
    Then $i_X\omega=0$ and thus, if we consider a basis of $T_xM$, $\{X,X_2,\ldots,X_{2n+p}\}$, we would have that $\omega^n\wedge\eta_1\wedge\cdots\wedge\eta_p(X,X_2,\ldots,X_{2n+p})=0$, which cannot be.
\end{proof}

\begin{corollary}
    Given an almost cosymplectic structure of order $p$ $(\omega,\eta_1,\cdots,\eta_p)$ on $M$, there exist unique vector fields $\mathcal{R}_k$, $k=1,\ldots,p$ such that:
    $$
    i_{\mathcal{R}_k}\omega=0 \qquad i_{\mathcal{R}_k}\eta_j=\delta_{kj}
    $$
The family $\{\mathcal{R}_k,\; k=1, \ldots, p\}$ will be called the family of Reeb vector field of the structure.
\end{corollary}
\begin{proof}
    It suffices to take $\flat^{-1}(\mathcal{R}_k)=\eta_k$.
\end{proof}

\begin{definition}
    Given $(\omega,\eta_1,\ldots,\eta_p)$ a partially cosymplectic structure of order $p$ over $M$, $F$ a 1-form on $M$ and $f$ a function over $M$, we define the evolution vector field of $f$ subject to the forces $F$ as the unique vector field over $M$ satisfying:
    $$
    \flat(E_f)=df+\eta-F
    $$
\end{definition}

\begin{definition}
    Let $(\omega,\eta_1,\ldots,\eta_p)$ and $(\Omega,\gamma_1,\ldots,\gamma_p)$ be partially cosymplectic structures of order $p$ over $M$ and $N$, respectively. A diffeomorphism $S: N \longrightarrow M$ is a cosymplectomorphism of order $p$ if:
    $$S^*\omega=\Omega \qquad S^*\eta_i = \gamma_i \qquad\, i=1,\ldots,p$$
\end{definition}

\begin{prop}
    Let $(\omega,\eta_1,\ldots,\eta_p)$ and $(\Omega,\gamma_1,\ldots,\gamma_p)$ be partially cosymplectic structures of order $p$ over $M$ and $N$ and $S: N \longrightarrow M$ a cosymplectomorphism of order $p$. Let $F$ be a 1-form on $M$ and $\tilde{F}=S^*F$ and let $f$ be a function on $M$, $\tilde{f}=f\circ S$. Let $E$ and $\xi$ be the evolution vector fields of f subject to the forces $F$ and $\tilde{f}$ subject to the forces $\tilde{F}$ on $M$ and $N$, respectively. Then:
    $$
    E=TS\circ \xi \circ S^{-1}
    $$
\end{prop}
\begin{proof}
    Let $X=TS \circ \xi \circ S^{-1}$
   $$
\flat(X)=i_X\omega+\sum_i\eta_i(X)\eta_i=i_X((S^{-1})^*\Omega)+\sum_i\left((S^{-1})^*\gamma_i\right)(X)(S^{-1})^*\gamma_i=$$ 
$$=(S^{-1})^*i_{\xi\circ S^{-1}}\Omega+\sum_i\gamma_i(\xi\circ S^{-1})(S^{-1})^*\gamma_i=(S^{-1})^*\left(i_{\xi\circ S^{-1}}\Omega+\sum_i\gamma_i(\xi\circ S^{-1})\gamma_i\right)
   $$ 
(and using the definition of $\xi$)
   $$
   =(S^{-1})^*\left(d\tilde{f}-\sum_i\gamma_i-\tilde{F}\right)=df-\sum_i\eta_i-F=\flat(E)
   $$
   As $\flat$ is an isomorphism, we conclude that both vector fields are equal.
\end{proof}

\section{A geometric description of thermodynamical systems}

\subsection{Adiabatically Closed Simple Thermodynamic Systems}
Let $Q$ be the configuration manifold that describes the mechanical variables of the system, and let $T^*Q$ be its cotangent bundle. Let the entropy of the system be described by a variable $S \in \mathbb{R}$. Let $M=T^*Q \times \mathbb{R}$.

Consider a Hamiltonian function:
$$
H: M \longrightarrow \mathbb{R}
$$
Let $F^{ext}, F^{fr}: M \longrightarrow T^*Q$ be  fiber-preserving functions which represent the external force and the friction force applied to the system. In local coordinates of $T^*Q$, $(q^i,p_i)$:
$$
F^{ext}=F^{ext}_i(q,p,S)dq^i, \qquad F^{fr}=F^{fr}_i(q,p,S)dq^i
$$
This is, $F^{ext},F^{fr}$ are semibasic forms. We define the 1-form over $M$:
$$
\eta=-\frac{\partial H}{\partial S}dS-F^{fr}
$$
and the 2-form over M:
$$
\omega = \pi_Q^*\omega_Q ,
$$
where $\pi_Q:M\longrightarrow T^*Q$ is the canonical projection and $\omega_Q$ is the canonical symplectic form defined over $T^*Q$. In local coordinates $(q^i,p_i,S)$:
$$
\omega=dq^i\wedge dp_i
$$
Notice that the pair $(\omega,\eta)$ defines a partially cosymplectic structure on M, regardless of what expression $F^{fr}$ takes.

Consider the isomorphism $\flat: TM \longrightarrow T^*M$ defined in Section 2.1.
\begin{remark}{\rm
    It is worth mentioning that although the isomorphism $\flat$ leads to the sum of two magnitudes with different dimensions, since $\omega$ has dimensions of action whereas $\eta\otimes\eta$ has dimensions of energy square, the corollary \ref{CotangentDecomposition} allows us to consider the decomposition $T^*M=H\oplus\langle\eta\rangle$ of the cotangent bundle. Thus, each of the terms of the sum that defines $\flat$ lies on a different vector bundle and there is no physical incompatibility in the sum.}
\end{remark}

Considering local coordinates we have that:
\begin{align*}
  \flat(\frac{\partial}{\partial q^i})&=dp_i-F^{fr}_i\eta \\
  \flat(\frac{\partial}{\partial p_i})&=-dq^i \\
  \flat(\frac{\partial}{\partial S})&=-\frac{\partial H}{\partial S}\eta 
\end{align*}
Let ${\mathcal E}_H$ be the evolution vector field of H subject to external forces, defined by the relation:
\begin{equation} \label{EvolutionFieldAdiabaticallyClosedSimple}
    \flat({\mathcal E}_H)=dH +\eta -F^{ext}
\end{equation}
Let ${\mathcal E}_H$ be locally given by:
$$
{\mathcal E}_H=A^i\frac{\partial}{\partial q^i}+B_i\frac{\partial}{\partial p_i}+C\frac{\partial}{\partial S}
$$
The right-hand side of (\ref{EvolutionFieldAdiabaticallyClosedSimple}) is locally given by:
\begin{equation} \label{RightEvolutionFieldAdiabaticallyClosedSimple}
    dH + \eta -F^{ext} = \left(\frac{\partial H}{\partial q^i}-F^{fr}_i-F^{ext}_i\right)dq^i+\frac{\partial H}{\partial p_i}dp_i
\end{equation}
Using the linearity of $\flat$, we conclude that, locally:
\begin{equation}\label{LeftEvolutionFieldAdiabaticallyClosedSimple}
    \flat({\mathcal E}_H)=-B_idq^i+A^idp_i-\left(A^iF^{fr}_i+C\frac{\partial H}{\partial S}\right) \eta
\end{equation}
Taking into account that the 1-forms $\{ dq^i,dp_i,\eta\}$ form a base of the cotangent space at the points of an open subset of $M$, we conclude that:
\begin{align*}
  &A^i =\frac{\partial H}{\partial p_i} \\
  &B_i=-\frac{\partial H}{\partial q^i}+F^{fr}_i +F^{ext}_i\\
  &A^iF^{fr}_i+C\frac{\partial H}{\partial S}=0 
\end{align*}
Thus, we have proved the following result:
\begin{prop}
Every integral curve of the evolution vector field ${\mathcal E}_H$ of $H$, $(q(t),p(t),S(t))$, is a solution of the equations:
\begin{align}
    \frac{d q^i}{dt} &= \frac{\partial H}{\partial p_i}\\
    \frac{d p_i}{dt} &= -\frac{\partial H}{\partial q^i}+F_i^{fr}+F^{ext}_i\label{CurveAdiabaticallyClosedSimple}\\
    \frac{d S}{dt} &= -\frac{1}{\frac{\partial H}{\partial S}}\frac{\partial H}{\partial p_j}F^{fr}_j
\end{align}
\end{prop}
In particular, we have the following equality
\begin{equation} \label{EntropyAdiabaticallyClosedSimple}
    -\frac{\partial H}{\partial S} \frac{d S}{dt}=\frac{d q^j}{dt}F^{fr}_j
\end{equation}
If we start by considering a regular Lagrangian function $L: TQ\times \mathbb{R} \longrightarrow \mathbb{R}$, then, as the Legendre transformation is a local diffeomorphism from the tangent bundle into the cotangent bundle, due to the properties of the product manifolds, it will be a local diffeomorphism when extended to an application from $TQ\times \R$ to $M$. 

If we define the energy of the Lagrangian as:
$$
E_L=\Delta (L) -L,
$$
where $\Delta$ is the Liouville vector field, we can define the Hamiltonian function locally as $H=E_L\circ Leg^{-1}$. In local coordinates:
$$
H(q,p,S)=p_i\Dot{q}^i(q,p,S)-L(q,\Dot{q}(q,p,S),S)
$$
Direct computation in local coordinates shows that:
$$
\frac{\partial H}{\partial S}=-\frac{\partial L}{\partial S}
$$
We define the 1-forms over $TQ\times\R$ given by $\tilde{F}^{ext}=Leg^*F^{ext}$ and $\tilde{F}^{fr}=Leg^*F^{fr}$. and the 1-form:
$$
\eta_L=\pder{L}{S}dS-\tilde{F}^{fr}
$$
We denote by
$$
\Omega_L=-d\lambda_L
$$
where $\lambda_L = S^*(dL)$ considered now as a 1-form on $TQ \times \mathbb{R}$. Then we have $Leg^*\omega= \Omega_L$ and $Leg^*\eta= \eta_L$. Thus, it is straightforward that $(\Omega_L,\eta_L)$ is a partially cosymplectic structure over $TM\times \R$. Let $\xi_L$ be the evolution vector field of $E_L$ subject to the external forces $\tilde{F}^{ext}$, defined by the relation:
$$
\flat({\mathcal E}_L)=dE_L+\eta_L-\tilde{F}^{ext}
$$
\begin{prop}
    If $L$ is hyperregular, that is, if the Legendre transformation $Leg$ is a global diffeomorphism, we can globally define $H$ on $M$ and:
    $$
    {\mathcal E}_H=TLeg \circ {\mathcal E}_L \circ Leg^{-1}
    $$
\end{prop}
\begin{proof}
It suffices to note that under these hypotheses, $Leg$ is a cosymplectomorphism of order $1$ and use the results of Section 3.2.
\end{proof}
\begin{corollary}
    Suppose $L$ is hyperregular. Then if $\gamma$ is an integral curve of ${\mathcal E}_L$, $\sigma=Leg^{-1} \circ \gamma$ is an integral curve of ${\mathcal E}_H$.    
\end{corollary}
\begin{corollary}
    If $L$ is a hyperregular Lagrangian and $\gamma$ is an integral curve of ${\mathcal E}_L$, it holds that:
    \begin{align} \label{CurveAdiabaticallyClosedSimpleLagrangian}
        &\frac{d}{dt}\left(\pder{L}{\Dot{q}^i}\right)-\pder{L}{q^i}=\tilde{F}_i^{fr}+\tilde{F}_i^{ext} \\
        &\pder{L}{S}\frac{dS}{dt}=\frac{dq^i}{dt}\tilde{F}^{fr}_i
    \end{align}
    where $\tilde{F}_i^{fr}=F_i^{fr}\circ Leg$ and $\tilde{F}_i^{ext}=F_i^{ext}\circ Leg$.
\end{corollary}
So we can conclude that equations (\ref{CurveAdiabaticallyClosedSimpleLagrangian}) are equivalent to those obtained by Gay-Balmaz and Yoshimura \cite{hiro}.
\begin{remark}{\rm 
    If $L$ is not hyperregular it suffices to work in an open subset of $TQ\times \R$ such that the restriction of $Leg$ to that subset is a diffeomorphism.}
\end{remark}
\begin{example}
    Previous studies of thermodynamics have been carried out using contact geometry. In \cite{simoes2020contact} an evolution field $\mathcal{E}_H$ is defined as:
    $$
    \flat({\mathcal E}_H)= dH-\mathcal{R}(H)\eta ,
    $$
    where $\mathcal{R}(H)=\pder{H}{S}$ and using local coordinates, $\eta=dS-p_idq^i$. Thus, considering a friction force locally given by $F^{fr}_i=-\mathcal{R}(H)p_i$ then the evolution vector defined in \cite{simoes2020contact} is an example of the one defined in this paper.
\end{example}

\subsection{Systems with Internal Mass Transfer}

Let Q be the configuration manifold that describes the mechanical part of the thermodynamic system and let $T^*Q$ be its cotangent bundle. 
We will consider a system with $K$ internal compartments, each of them with $N_k$ particles. We will denote $P_2=\R^K$ the manifold that represents the number of particles in each compartment. We will also consider in each compartment the thermodynamic displacement associated with the exchange of matter, $W^k$ (whose derivative over the trajectory will be the temperature). Let $P_1=\R^K$ be the manifold that represents these thermodynamic displacements. Finally, we will consider the entropy of the system described by a real variable, $S\in \R$.

Let $N=T^*Q\times P_2\times\R$ and $M=T^*Q\times P_1\times P_2\times\R$, and let $\pi:M \longrightarrow N$ be the canonical projection. We will consider a Hamiltonian function independent of the thermodynamic displacement, that is, a function $\Tilde{H}:N\longrightarrow\R$ and its pullback by $\pi$, $H=\pi^*\Tilde{H}$.

We will also consider the 1-forms $\hat{F}^{fr},\hat{F}^{ext}:N\longrightarrow T^*Q$ and their pullbacks by $\pi$, $F^{fr}$, $F^{ext}$, which represent the friction and the external forces acting on the system. If we choose local adapted coordinates $(q^i,p_i)$ in $T^*Q$:
$$
F^{fr}=F^{fr}_idq^i\qquad F^{ext}=F^{ext}_idq^i
$$
We consider functions $\mathcal{J}_{l,k}:N\longrightarrow\R$ such that $\mathcal{J}_{l,k}=-\mathcal{J}_{k,l}$, which we will identify with their pullbacks by $\pi$. Let:
$$
\mathcal{J}=\sum_l \mathcal{J}_{l,k}dW^k:=\mathcal{J}_kdW^k
$$
Consider in $M$ the 1-form given by:
$$
\eta=-\frac{\partial H}{\partial S}dS-F^{fr}-\mathcal{J}
$$
and the 2-form:
$$
\omega = dq^i\wedge dp_i+dW^k\wedge dN_k
$$
Then $(\omega,\eta)$ is a partially cosymplectic structure over $M$. Considering local coordinates $(q^i,p_i)$ in $T^*Q$, the isomorphism $\flat$ satisfies:
\begin{align*}
    &\flat\left( \frac{\partial}{\partial q^i}\right)=dp_i-F_i^{fr}\eta &\qquad& \flat\left( \frac{\partial}{\partial W^k}\right)=dN_k-\mathcal{J}_k\eta \\
    &\flat\left( \frac{\partial}{\partial p_i}\right)=-dp^i &\qquad& \flat\left( \frac{\partial}{\partial N_k}\right)=-dW^k\\
   &\flat\left( \frac{\partial}{\partial S}\right)=-\frac{\partial H}{\partial S}\eta  
\end{align*}
Let ${\mathcal E}_H$ be the evolution vector field of $H$ subject to external forces, defined by the relation:
\begin{equation} \label{EvolutionFieldMassSimple}
    \flat({\mathcal E}_H)=dH +\eta -F^{ext}
\end{equation}
Using local coordinates, the expression of the right-hand side of (\ref{EvolutionFieldMassSimple}) is:
\begin{equation} \label{RightEvolutionFieldMassSimple}
    dH + \eta -F^{ext}= \left(\frac{\partial H}{\partial q^i}-F^{fr}_i-F^{ext}_i\right)dq^i+\frac{\partial H}{\partial p_i}dp_i+\frac{\partial H}{\partial N_k}dN_k-\mathcal{J}_kdW^k
\end{equation}
Setting ${\mathcal E}_H=A^i\frac{\partial}{\partial q^i}+B_i\pder{ }{p_i}+C^k\pder{ }{W^k}+D_k\pder{ }{N_k}+E\pder{}{S}$ and using the linearity of $\flat$, we conclude that, locally:
\begin{equation}\label{LeftEvolutionFieldMassSimple}
    \flat({\mathcal E}_H)=-B_idq^i+A^idp_i- D_k dW^k + C^kdN_k -\left(A^iF^{fr}_i+C^k\mathcal{J}_k+E\frac{\partial H}{\partial S}\right) \eta
\end{equation}
Taking into account that the 1-forms $\{dq^i,dp_i,dW^k,dN_k,\eta\}$ form a basis of $T^*M$ at the points of an open subset of $M$, we can equal the coefficients of these 1-forms at both sides of equation (\ref{EvolutionFieldMassSimple}). This proves the following result: 
\begin{prop}
Every integral curve of the evolution vector field ${\mathcal E}_H$ of $H$, $\sigma(t)$, is a solution of the equations:
\begin{align}
     & \frac{d q^i}{dt} = \frac{\partial H}{\partial p_i}\\
    &\frac{d p_i}{dt} = -\frac{\partial H}{\partial q^i}+F_i^{fr}+F^{ext}_i\\
    &\frac{d W^k}{dt} = \frac{\partial H}{\partial N_k}\\
    &\frac{d N_k}{dt} = \mathcal{J}_k\\
    &\frac{d S}{dt} = -\frac{1}{\frac{\partial H}{\partial S}}\left(\frac{\partial H}{\partial p_j}F^{fr}_j+\mathcal{J}_k\pder{H}{N_k}\right)
\end{align}
\end{prop}
In particular, we have the following.
\begin{equation} \label{EntropyMassdSimple}
    -\frac{\partial H}{\partial S} \frac{d S}{dt}=\frac{d q^j}{dt}F^{fr}_j+\mathcal{J}_k\pder{H}{N_k}
\end{equation}

As in the previous section, if we consider a regular Lagrangian function:
$$
L: TQ\times P_2 \times \R\longrightarrow \R
$$
or more precisely, its pullback to $D=TQ\times P_1 \times P_2 \times \R$ then, as the Legendre transform is a local diffeomorphism from the tangent bundle into the cotangent bundle, due to the properties of the product manifolds, it will be a local diffeomorphism when extended to an application from $D$ to $M$. 

If we define the energy of the Lagrangian as:
$$
E_L=\Delta (L) -L
$$
where $\Delta$ is the Liouville vector field, then we can define the Hamiltonian function locally as $H=E_L\circ Leg^{-1}$. Direct computation in local coordinates again shows that:
$$
\frac{\partial H}{\partial S}=-\frac{\partial L}{\partial S}
$$
We define the 1-forms over $D$ given by $\tilde{F}^{ext}=Leg^*F^{ext}$, $\tilde{F}^{fr}=Leg^*F^{fr}$ and $\tilde{\mathcal{J}}=Leg^*\mathcal{J}$. and the 1-form:
$$
\eta_L=\pder{L}{S}dS-\tilde{F}^{fr}-\tilde{\mathcal{J}}
$$
We denote by
$$
\Omega_L=-d\lambda_L+dW^k\wedge dN_k
$$
where $\lambda_L = S^*(dL)$ considered now as a 1-form on $D$. Then we have $Leg^*\omega= \Omega_L$ and $Leg^*\eta= \eta_L$. Thus, it is straightforward that $(\Omega_L,\eta_L)$ is a partially cosymplectic structure over $D$. Let $\xi_L$ be the evolution vector field of $E_L$ subject to the external forces $\tilde{F}^{ext}$, defined by the relation:
$$
\flat({\mathcal E}_L)=dE_L+\eta_L-\tilde{F}^{ext}
$$
\begin{prop}
    If $L$ is hyperregular, that is, if the Legendre transformation $Leg$ is a global diffeomorphism, we can globally define $H$ on $M$ and:
    $$
    {\mathcal E}_H=TLeg \circ {\mathcal E} \circ Leg^{-1}
    $$
\end{prop}
\begin{proof}
It follows from the same reasoning as Proposition 4.2.
\end{proof}
\begin{corollary}
    Suppose $L$ is hyperregular. Then if $\gamma$ is an integral curve of ${\mathcal E}_L$, $\sigma=Leg^{-1} \circ \gamma$ is an integral curve of ${\mathcal E}_H$.    
\end{corollary}
\begin{corollary}
    If $L$ is a hyperregular Lagrangian and $\gamma$ is an integral curve of ${\mathcal E}_L$, it holds that:
    \begin{align} \label{CurveMassLagrangian}
        &\frac{d}{dt}\left(\pder{L}{\Dot{q}^i}\right)-\pder{L}{q^i}=\tilde{F}_i^{fr}+\tilde{F}_i^{ext} \\
        &\frac{d W^k}{dt} = -\frac{\partial L}{\partial N_k}\\
        &\frac{d N_k}{dt} = \tilde{\mathcal{J}}_k\\
    &\frac{\partial L}{\partial S}\frac{d S}{dt} = \frac{dq^i}{dt}F^{fr}_j-\tilde{\mathcal{J}}_k\pder{L}{N_k}
    \end{align}
    where $\tilde{F}_i^{fr}=F_i^{fr}\circ Leg$, $\tilde{F}_i^{ext}=F_i^{ext}\circ Leg$ and $\tilde{\mathcal{J}}_k=\mathcal{J}_k\circ Leg$.
\end{corollary}
So we can conclude that equations (\ref{CurveMassLagrangian}) are equivalent to those obtained by Gay-Balmaz and Yoshimura \cite{hiro}.

\subsection{Adiabatically Closed Non-Simple Thermodynamic Systems}

Consider a thermodynamical system composed of $P$ simple subsystems, each of them characterized by their entropy $S_A$. Consider that heat conduction, friction, and internal mass transfer occur. We will restrict our study to the case in which each subsystem has only one compartment.

Consider that the mechanical variables that describe the entire system lie in a manifold $Q$. Let $T^*Q$ be its cotangent bundle. Let the thermodynamical displacements associated with mass transfer be described by $P_1=\R^P$ and the number of particles in each subsystem by $P_2=\R^P$.

Similarly, let $P_3=\R^P$ describe thermic displacements (whose time derivative over the trajectory will be the temperature), $\Gamma^A$, and let $P^4$ describe the entropies of the subsystems, $S_A$. 
We will need to consider an auxiliary variable $\Sigma_A$ for each subsystem, which will equal the entropy of the subsystem on the trajectory of the system. Let these variables be in $P_5=\R^P$.

Let $N= T^*Q\times P_2 \times P_4$ and $M=T^*Q\times P_1 \times P_2 \times P_3 \times P_4 \times P_5$ and the canonical projection $\pi:M\longrightarrow N$. 

We consider 1-forms $F_A^{fr},F_A^{ext}:N\longrightarrow T^*Q$, which we identify with their pullbacks by $\pi$. If we choose local adapted coordinates $(q^i,p_i)$ in $T^*Q$:
$$
F_A^{fr}=F_{A,i}^{fr}dq^i\qquad F_A^{ext}=F_{A,i}^{ext}dq^i
$$
Let $F^{fr}=\sum_AF^{fr}_A$ and  $F^{ext}=\sum_AF^{ext}_A$.

We also consider a Hamiltonian function $H:N\longrightarrow \R$ that we identify with its pullback by $\pi$. Similarly, we consider functions $\mathcal{J}_{l,k}:N\longrightarrow\R$ such that $\mathcal{J}_{l,k}=-\mathcal{J}_{k,l}$ and again identify them with their pullbakcs to $M$. Let $\mathcal{J}_k=\sum_l\mathcal{J}_{l,k}$. Finally, we consider functions $J_{AB}:N\longrightarrow\R$ such that $\sum_AJ_{AB}=0$ and identify them with their pullback to M.

We define $P$ 1-forms $\eta_A \; ; A=1,\ldots, P$ and the 2-form $\omega$ as follows:
\begin{align*}
    &\eta_A = -\pder{H}{S_A}d\Sigma_A-F^{fr}_A-\mathcal{J}_AdW^A-J_{AB}d\Gamma^B \\
    &\omega=dq^i\wedge dp_i+dW^k\wedge dN_k+d\Gamma^A\wedge d(S_A-\Sigma_A)
\end{align*}
Notice that, in the definition of $\eta_A$ we do not sum over $A$ in the first and third terms.

Then $(\omega, \eta_1,\cdots,\eta_P)$ is a partially cosymplectic structure of order $P$ over $M$. Consider the isomorphism $\flat$ defined over this structure and the evolution vector field with external forces defined by the relation:
\begin{equation}\label{EvolutionFieldNonSimple}
\flat({\mathcal E}_H)=dH+\sum_A \eta_A-F^{ext}
\end{equation}
Considering local coordinates of $T^*Q$, $(q^i,p_i)$, we may express the right-hand side of the previous equation locally as:
\begin{equation}\label{RightEvolutionFieldNonSimple}
dH+\sum_A \eta_A-F^{ext}=\left(\pder{H}{q^i}-F^{fr}_i-F^{ext}_i\right)dq^i+\pder{H}{p_i}dp_i-\mathcal{J}_kdW^k+\pder{H}{N_k}dN_k+\pder{H}{S_A}d(S_A-\Sigma_A)
\end{equation}
In these local coordinates the isomorphism $\flat$ satisfies:
\begin{align*}
    &\flat\left( \frac{\partial}{\partial q^i}\right)=dp_i-\sum_k F_{k,i}^{fr}\eta_k &\qquad& \flat\left( \frac{\partial}{\partial p_i}\right)=-dp^i \\
    &\flat\left( \frac{\partial}{\partial W^k}\right)=dN_k-\mathcal{J}_k\eta_k &\qquad& \flat\left( \frac{\partial}{\partial N_k}\right)=-dW^k\\
    &\flat\left( \frac{\partial}{\partial \Gamma^A}\right)=d(S_A-\Sigma_A)-\sum_kJ_{kA}\eta_k &\qquad& \flat\left( \frac{\partial}{\partial S_A}\right)=-d\Gamma^A\\
   &\flat\left( \frac{\partial}{\partial \Sigma_A}\right)=d\Gamma^A-\frac{\partial H}{\partial S_A}\eta_A  
\end{align*}
Setting:
$$
{\mathcal E}_H=A^i\pder{}{q^i}+B_i\pder{}{p_i}+C^k\pder{}{W^k}+D_k\pder{}{N_k}+E^A\pder{}{\Gamma^A}+F_A\pder{}{S_A}+G_A\pder{}{\Sigma_A}
$$
and using the linearity of $\flat$, we have that:
\begin{multline}\label{LeftEvolutionFieldNonSimple}
\flat({\mathcal E}_H)=A^idp_i-B_idq^i+C^kdN_k-D_kdW^k+E^Ad(S_A-\Sigma_A)+ \\
+(G_A-F_A)d\Gamma^A-\sum_k\left(A^iF_{k,i}^{fr}+C^k\mathcal{J}_k+E^AJ_{kA}+G_k\pder{H}{S_k}\right)\eta_k
\end{multline}
Since $\{dq^i,dp_i,dW^k,dN_k,d\Gamma^A,d(S_A-\Sigma_A),\eta_A\}$ form a basis of $T^*M$ at each point of an open subset of $M$ we can equal the coefficients of these 1-forms at both sides of equation (\ref{EvolutionFieldNonSimple}) and conclude the following result:
\begin{prop}
Every integral curve of ${\mathcal E}_H$, $\sigma(t)=(q(t),p(t),W(t),N(t),\Gamma(t),S(t),\Sigma(t))$, is a solution of the equations:
\begin{align}
     & \frac{d q^i}{dt} = \frac{\partial H}{\partial p_i}\\
    &\frac{d p_i}{dt} = -\frac{\partial H}{\partial q^i}+\sum_k F_{k,i}^{fr}+\sum_k F^{ext}_{k,i}\\
    &\frac{d W^k}{dt} = \frac{\partial H}{\partial N_k}\\
    &\frac{d N_k}{dt} = \mathcal{J}_k\\
    & \frac{d \Gamma^A}{dt} = \frac{\partial H}{\partial S_A}\\
    & \frac{d S_A}{dt} = \frac{d \Sigma_A}{dt}\\
    &\frac{d S_k}{dt} = -\frac{1}{\frac{\partial H}{\partial S}}\left(\frac{\partial H}{\partial p_j}F^{fr}_{k,j}+\mathcal{J}_k\pder{H}{N_k}+\pder{H}{S_A}J_{kA}\right)
\end{align}
\end{prop}
In particular, denoting $\displaystyle{T^A=\pder{H}{S_A}}$ and $\displaystyle{\mu^k=\pder{H}{N_k}}$ and taking into account that $\sum_AJ_{kA}T^k=0$, we have the following:
\begin{equation} \label{EntropyNonSimple}
    -T^k \frac{d S_k}{dt}=\frac{d q^j}{dt}F^{fr}_{k,j}+\sum_AJ_{kA}(T^A-T^k)+\mathcal{J}_k\mu^k
\end{equation}

As in previous sections, if we consider a regular Lagrangian function:
$$
L: TQ\times P_2 \times P_4 \longrightarrow \R
$$
or more precisely, its pullback to $D=TQ\times P_1 \times P_2 \times P_3\times P_4 \times P_5$ then, as the Legendre transformation is a local diffeomorphism, due to the properties of the product manifolds, it will be a local diffeomorphism when extended it to an application from $D$ to $M$. 

We define the energy of the Lagrangian as:
$$
E_L=\Delta (L) -L
$$
where $\Delta$ is the Liouville vector field, we can define the Hamiltonian function locally as $H=E_L\circ Leg^{-1}$. Direct computation in local coordinates again shows that:
$$
\frac{\partial H}{\partial S_A}=-\frac{\partial L}{\partial S_A}  \qquad A=1,\ldots, P
$$
We define the 1-forms over $D$ given by $\tilde{F}^{ext}_A=Leg^*F^{ext}_A$ and $\tilde{F}^{fr}_A=Leg^*F^{fr}_A$ as well as the functions in $D$ given by $\tilde{J}_{AB}=J_{AB}\circ Leg$ and $\tilde{\mathcal{J}}_k=\mathcal{J}_k\circ Leg$. Let the 1-forms $\eta_{A,L}$ be:
$$
\eta_{A,L}=\pder{L}{S_A}d\Sigma_A-\tilde{F}^{fr}_A-\tilde{\mathcal{J}}_AdW^A-\tilde{J}_{AB}d\Gamma^B
$$
Note that, as in the definition of $\eta_A$, we do not sum over $A$ in the first and third term of the right-hand side.
We denote by
$$
\Omega_L=-d\lambda_L+dW^k\wedge dN_k+d\Gamma^A\wedge d(S_A-\Sigma_A)
$$
where $\lambda_L = S^*(dL)$ is considered now as a 1-form on $D$. Then we have $Leg^*\omega= \Omega_L$ and $Leg^*\eta_A= \eta_{A,L}$. Thus, it is straightforward that $(\Omega_L,\eta_{1,L},\cdots,\eta_{P,L})$ is a partially cosymplectic structure of order $P$ over $D$. Let ${\mathcal E}_L$ be the evolution vector field of $E_L$ subject to the external forces $\tilde{F}^{ext}$, defined by the relation:
$$
\flat({\mathcal E}_L)=dE_L+\sum_A\eta_{A,L}-\tilde{F}^{ext}
$$
where $\tilde{F}^{ext}=\sum_A\tilde{F}^{ext}_A$.
\begin{prop}
    If $L$ is hyperregular, that is, if the Legendre transformation $Leg$ is a global diffeomorphism, we can globally define $H$ on $M$ and:
    $$
    {\mathcal E}_H=TLeg \circ {\mathcal E}_L \circ Leg^{-1}
    $$
\end{prop}
\begin{proof}
Notice that under these hypotheses, $Leg$ is a cosymplectomorphism of order $P$.
\end{proof}
\begin{corollary}
    Suppose $L$ is hyperregular. Then if $\gamma$ is an integral curve of $\xi_L$, $\sigma=Leg^{-1} \circ \gamma$ is an integral curve of ${\mathcal E}_H$.    
\end{corollary}
\begin{corollary}
    $L$ is a hyperregular Lagrangian and $\gamma$ is an integral curve of ${\mathcal E}_L$, it holds that:
    \begin{align} \label{CurveNonSimpleLagrangian}
        &\frac{d}{dt}\left(\pder{L}{\Dot{q}^i}\right)-\pder{L}{q^i}=\sum_k\tilde{F}_{k,i}^{fr}+\sum_k\tilde{F}_{k,i}^{ext} \\
        &\frac{d W^k}{dt} = -\frac{\partial L}{\partial N_k}\\
        &\frac{d N_k}{dt} = \tilde{\mathcal{J}}_k\\
        & \frac{d \Gamma^A}{dt} = -\frac{\partial L}{\partial S_A}\\
        & \frac{d S_A}{dt} = \frac{d \Sigma_A}{dt}\\
        &\frac{\partial L}{\partial S}\frac{d S_k}{dt} =\frac{d q^j}{dt}\tilde{F}^{fr}_{k,j}-\tilde{\mathcal{J}}_k\pder{L}{N_k}-\pder{L}{S_A}J_{kA}
    \end{align}
    where $\tilde{F}_{k,i}^{fr}=F_{k,i}^{fr}\circ Leg$ and $\tilde{F}_{k,i}^{ext}=F_{k,i}^{ext}\circ Leg$. 
\end{corollary}
So we can conclude that equations (\ref{CurveNonSimpleLagrangian}) are equivalent to those obtained by Gay-Balmaz and Yoshimura \cite{hiro} for closed non-simple thermodynamical systems.

\subsection{Open Simple Thermodynamic System}

Let us consider, as in \cite{hiro}, an open simple thermodynamic system with only one chemical species and one compartment. We denote by $N$ the number of moles of this species. Consider that this system is in contact with the exterior at several ports, $a=1,\ldots,A$, which allow the flow of matter, and with several heat sources $b=1,\ldots,B$.

Let $Q$ be the manifold describing the mechanical part of our system and $T^*Q$ its cotangent bundle. Let $P_1,P_2,P_3,P_4,P_5$ be defined as in the previous section, taking $P=1$ since we are dealing with a simple system. Let $N$, $M$ be defined as in the previous section as well as the 1-forms $F^{fr},F^{ext}$ which account for the external force and the dissipative force acting on the system. 

We also define the functions $\mathcal{J}^a:N\longrightarrow\R$ and identify them with their pullbacks to $M$. These will be the molar flow rate into the system through the $a$-th port. Similarly, we define $\mu^a, T^a, T^b, J^b_S, S^a:N\longrightarrow\R$ and identify them with their pullbacks. They will respectively represent the chemical potential at the $a$-th port, the temperature at the $a$-th port, the temperature of the $b$-th heat source, the entropy flow rate into the system from the $b$-th heat source and the molar entropy at the $a$-th port. We finally define $\mathcal{J}^a_S=\mathcal{J}^aS^a$ as the entropy flow rate into the system at the $a$-th port.

We define the 1-form $\eta$ and the 2-form $\omega$ as follows:
\begin{align*}
    &\eta = -\pder{H}{S}d\Sigma-F^{fr}-\sum_{a=1}^A\left(\mathcal{J}^adW+\mathcal{J}^a_Sd\Gamma\right)-\sum_{b=1}^BJ^b_Sd\Gamma \\
    &\omega=dq^i\wedge dp_i+dW\wedge dN+d\Gamma\wedge d(S-\Sigma)
\end{align*}
Then we have that $(\omega,\eta)$ is a partially cosymplectic structure on $M$. Consider $\flat:TM\longrightarrow T^*M$ the canonical isomorphism corresponding to this structure. We define the evolution vector field with external forces as the vector field ${\mathcal E}_H$ satisfying:
\begin{equation}\label{CampoEvolucionAbierto}
    \flat({\mathcal E}_H)=dH+\eta-F^{ext}-\left(\sum_{a=1}^A(\mathcal{J}^a\mu^a+\mathcal{J}^a_ST^a )+ \sum_{b=1}^BJ^b_ST^b\right)\eta
\end{equation}

Let $(q^i,p_i)$ be adapted coordinates in $T^*Q$ and consider the local coordinates $(q^i,p_i,W,N,\Gamma,S,\Sigma)$. Then the right-hand side of (\ref{CampoEvolucionAbierto}) may be expressed in local coordinates as:
\begin{multline}\label{RightEvolutionFieldOpen}
    \left(\pder{H}{q^i}-F^{fr}_i-F^{ext}_i\right)dq^i+\pder{H}{p_i}dp_i-\sum_{a=1}^A\mathcal{J}^adW+\pder{H}{N}dN-\\ 
-\left(\sum_{a=1}^A\mathcal{J}^a_S+\sum_{b=1}^BJ^b_S\right)d\Gamma+\pder{H}{S}d(S-\Sigma)-\left(\sum_{a=1}^A(\mathcal{J}^a\mu^a+\mathcal{J}^a_ST^a )+ \sum_{b=1}^B J^b_ST^b\right)\eta
\end{multline}
In these local coordinates, the isomorphism $\flat$ satisfies:
\begin{align*}
    &\flat\left( \frac{\partial}{\partial q^i}\right)=dp_i- F_{i}^{fr}\eta &\qquad& \flat\left( \frac{\partial}{\partial p_i}\right)=-dp^i \\
    &\flat\left( \frac{\partial}{\partial W}\right)=dN-\sum_{a=1}^A\mathcal{J}^a\eta &\qquad& \flat\left( \frac{\partial}{\partial N}\right)=-dW\\
    &\flat\left( \frac{\partial}{\partial \Gamma}\right)=d(S-\Sigma)-\left(\sum_{a=1}^A\mathcal{J}^a_S+\sum_{b=1}^BJ^b_S\right)\eta &\qquad& \flat\left( \frac{\partial}{\partial S}\right)=-d\Gamma\\
   &\flat\left( \frac{\partial}{\partial \Sigma}\right)=d\Gamma-\frac{\partial H}{\partial S}\eta  
\end{align*}
Setting:
$$
{\mathcal E}_H=A^i\pder{}{q^i}+B_i\pder{}{p_i}+C\pder{}{W}+D\pder{}{N}+E\pder{}{\Gamma}+F\pder{}{S}+G\pder{}{\Sigma}
$$
And using the linearity of $\flat$, we have that:
\begin{multline}\label{LeftEvolutionFieldOpen}
\flat({\mathcal E}_H)=A^idp_i-B_idq^i+CdN-DdW+Ed(S-\Sigma)+ \\
+(G-F)d\Gamma-\left(A^iF_{i}^{fr}+\sum_{a=1}^AC\mathcal{J}^a+\left(\sum_{a=1}^A\mathcal{J}^a_S+\sum_{b=1}^BJ^b_S\right)E+G\pder{H}{S}\right)\eta
\end{multline}
Using that $\{dq^i,dp_i,dN, dW, d(S-\Sigma),d\Gamma,\eta\}$ span a basis at each point of an open subset of $M$ of $T^*_xM$, we can conclude the following result.
\begin{prop}
Every integral curve of ${\mathcal E}_H$, $\sigma(t)=(q(t),p(t),W(t),N(t),\Gamma(t),S(t),\Sigma(t))$, is a solution of the equations:
\begin{align}
     & \frac{d q^i}{dt} = \frac{\partial H}{\partial p_i}\\
    &\frac{d p_i}{dt} = -\frac{\partial H}{\partial q^i}+ F_{i}^{fr}+ F^{ext}_{i}\\
    &\frac{d W}{dt} = \frac{\partial H}{\partial N}\\
    &\frac{d N}{dt} = \sum_{a=1}^A\mathcal{J}^a\\
    & \frac{d \Gamma}{dt} = \frac{\partial H}{\partial S}\\
    & \frac{d S}{dt} = \frac{d \Sigma}{dt}+\left(\sum_{a=1}^A\mathcal{J}^a_S+\sum_{b=1}^BJ^b_S\right)\\
    &-\frac{\partial H}{\partial S}\frac{d \Sigma}{dt} = \frac{dq^i}{dt}F^{fr}_{i}+\sum_{a=1}^A\mathcal{J}^a\frac{d W}{dt}+\left(\sum_{a=1}^A\mathcal{J}^a_S+\sum_{b=1}^BJ^b_S\right)\frac{d \Gamma}{dt}- \\
    &\qquad\qquad\quad \;\,-\left(\sum_{a=1}^A(\mathcal{J}^a\mu^a+\mathcal{J}^a_ST^a )+ \sum_{b=1}^B J^b_ST^b\right) \nonumber
\end{align}
\end{prop}

As in previous sections, if we consider a regular Lagrangian function:
$$
L: TQ\times P_2 \times P_4 \longrightarrow \R
$$
or, more precisely, its pullback to $D=TQ\times P_1 \times P_2 \times P_3\times P_4 \times P_5$, then, as the Legendre transformation is a local diffeomorphism, due to the properties of the product manifolds it will be a local diffeomorphism when extended to an application from $D$ to $M$. 

If we define the energy of the Lagrangian as:
$$
E_L=\Delta (L) -L
$$
where $\Delta$ is the Liouville vector field, then we can define the Hamiltonian function locally as $H=E_L\circ Leg^{-1}$. Direct computation in local coordinates again shows that:
$$
\frac{\partial H}{\partial S}=-\frac{\partial L}{\partial S}
$$
We define the 1-forms over $D$ given by $\tilde{F}^{ext}=Leg^*F^{ext}$ and $\tilde{F}^{fr}=Leg^*F^{fr}$ as well as the functions in $D$ $\tilde{J}_S^b$, $\tilde{\mathcal{J}}^a$, $\tilde{\mathcal{J}}^a_S$, $\tilde{\mu}^a$, $\tilde{T}^a$ and $\tilde{T}^b$ defined as the composition of the respective functions defined on M without a tilde composed with $Leg$. Let the 1-form $\eta_{L}$ be:
$$
\eta_{L}=\pder{L}{S}d\Sigma-\tilde{F}^{fr}-\sum_{a=1}^A(\tilde{\mathcal{J}}^adW+\tilde{\mathcal{J}}^a_Sd\Gamma)-\sum_{b=1}^B\tilde{J}_{S}^bd\Gamma
$$
We denote by
$$
\Omega_L=-d\lambda_L+dW\wedge dN+d\Gamma\wedge d(S-\Sigma)
$$
where $\lambda_L = S^*(dL)$ considered now as a 1-form on $D$. Then we have $Leg^*\omega= \Omega_L$ and $Leg^*\eta= \eta_{L}$. Thus, it is straightforward that $(\Omega_L,\eta_{L})$ is a partially cosymplectic structure over $D$. Let ${\mathcal E}_L$ be the evolution vector field of $E_L$ subject to the external forces $\tilde{F}^{ext}+\left(\sum_{a=1}^A(\tilde{\mathcal{J}}^a\mu^a+\tilde{\mathcal{J}}^a_S\tilde{T}^a )+ \sum_{b=1}^B\tilde{J}^b_S\tilde{T}^b\right)\eta_L$, defined by the relation:
$$
\flat({\mathcal E}_L)=dE_L+\eta_{L}-\tilde{F}^{ext}-\left(\sum_{a=1}^A(\tilde{\mathcal{J}}^a\mu^a+\tilde{\mathcal{J}}^a_S\tilde{T}^a )+ \sum_{b=1}^B\tilde{J}^b_S\tilde{T}^b\right)\eta_L
$$
\begin{prop}
    If $L$ is hyperregular, that is, if the Legendre transformation $Leg$ is a global diffeomorphism, we can globally define $H$ on $M$ and:
    $$
    {\mathcal E}_H=TLeg \circ {\mathcal E}_L \circ Leg^{-1}
    $$
\end{prop}
\begin{proof}
Notice that in this situation, $Leg$ is a cosymplectomorphism subject to external forces.
\end{proof}              
\begin{corollary}
    Suppose $L$ is hyperregular. Then if $\gamma$ is an integral curve of ${\mathcal E}_L$, $\sigma=Leg^{-1} \circ \gamma$ is an integral curve of ${\mathcal E}_H$.    
\end{corollary}
\begin{corollary}
    $L$ is a hyperregular Lagrangian and $\gamma$ is an integral curve of ${\mathcal E}_L$, it holds that:
    \begin{align} \label{CurveOpenLagrangian}
        &\frac{d}{dt}\left(\pder{L}{\Dot{q}^i}\right)-\pder{L}{q^i}=\tilde{F}_{i}^{fr}+\tilde{F}_{i}^{ext} \\
        &\frac{d W}{dt} = -\frac{\partial L}{\partial N}\\
        &\frac{d N}{dt} = \sum_{a=1}^A\tilde{\mathcal{J}}^a\\
        & \frac{d \Gamma}{dt} = -\frac{\partial L}{\partial S}\\
        & \frac{d S}{dt} = \frac{d \Sigma}{dt}+\left(\sum_{a=1}^A\tilde{\mathcal{J}}^a_S+\sum_{b=1}^B\tilde{J}^b_S\right)\\
        &\frac{\partial L}{\partial S}\frac{d \Sigma}{dt} = \frac{dq^i}{dt}\tilde{F}^{fr}_{i}+\sum_{a=1}^A\tilde{\mathcal{J}}^a\frac{d W}{dt}+\left(\sum_{a=1}^A\tilde{\mathcal{J}}^a_S+\sum_{b=1}^B\tilde{J}^b_S\right)\frac{d \Gamma}{dt}-\\
        &\qquad\qquad -\left(\sum_{a=1}^A(\tilde{\mathcal{J}}^a\mu^a+\tilde{\mathcal{J}}^a_S\tilde{T}^a )+ \sum_{b=1}^B \tilde{J}^b_S\tilde{T}^b\right)
    \end{align}
    where $\tilde{F}_{i}^{fr}=F_{i}^{fr}\circ Leg$ and $\tilde{F}_{i}^{ext}=F_{i}^{ext}\circ Leg$. 
\end{corollary}
So we can conclude that equations (\ref{CurveOpenLagrangian}) are equivalent to those obtained by Gay-Balmaz and Yoshimura \cite{hiro} for open thermodynamical systems.

\section{Conclusions and further work}

In this paper we have introduced several new geometric structures which are a natural setting to describe a wide variety of thermodynamical systems. Indeed, we are able to obtain the same
evolution equations obtained previously by Gay-Balmaz and Yoshimura \cite{hiro} using variational arguments. This new mathematical description will be used in forthcoming research to discuss several items.

\begin{itemize}

\item Develop a deeper study on these geometries identifying them as $G$-structures.

\item Study their submanifolds trying to obtain notions equivalent to the usual Lagrangian, cosisotropic or isotropic submanifolds in the symplectic and cosymplectic setting.

\item Identify the corresponding almost Poisson brackets associated to these geometric structures, and use them to describe the dynamics. Typically, in the previous literature, one of the most successful methods are based on the introduction of metriplectic structures,
coupling a Poisson and a gradient structure, where the entropy S is now constructed from a Casimir function of the Poisson structure \cite{uno,dos}. We want to go into the relations with our approach.

\item Study the reduction procedure under the existence of a Lie group of symmetries.

\item Obtain a convenient Hamilton-Jacobi theory in these new settings.

\item Develop discretization processes as in the case of symplectic and contact systems \cite{symplectic,mio,brav}.
\end{itemize}

\section*{Acknowledgments}

We acknowledge financial support of the 
{\sl Ministerio de Ciencia, Innovaci\'on y Universidades} (Spain), grants PID2022-137909NB-C21 and RED2022-134301-T.
We also acknowledge financial support from the Severo Ochoa Programme for Centers of Excellence in R\&D.

\phantomsection
\addcontentsline{toc}{section}{References}
\pagestyle{empty}
\bibliographystyle{plain} 
\bibliography{refs} 
\end{document}